\documentclass[11pt]{article}

\usepackage[letterpaper,margin=1.00in]{geometry}
\usepackage{amsmath, amssymb, amsthm, amsfonts}
\usepackage{bbm}

\usepackage{ifthen}
\usepackage{tikz}
\usetikzlibrary{positioning,decorations.pathreplacing}
\usepackage{ bbold }
\usepackage{cite}
\usepackage{appendix}
\usepackage{graphicx}
\usepackage{color}
\usepackage{color}
\usepackage{algorithm}
\usepackage{algorithm}
\usepackage[noend]{algpseudocode}
\usepackage{epstopdf}
\usepackage{wrapfig}
\usepackage{paralist}
\usepackage{wasysym}
\usepackage[textsize=tiny]{todonotes}

\usepackage{framed}
\usepackage[framemethod=tikz]{mdframed}
\usepackage[bottom]{footmisc}
\usepackage{enumitem}
\setitemize{noitemsep,topsep=3pt,parsep=3pt,partopsep=3pt}
\usepackage[font=small]{caption}
\usepackage{xspace}
\usepackage{thmtools}
\usepackage{thm-restate}
\usepackage{multirow}
\usepackage{array}

\newtheorem{theorem}{Theorem}[section]
\newtheorem{lemma}[theorem]{Lemma}
\newtheorem{meta-theorem}[theorem]{Meta-Theorem}

\definecolor{darkgreen}{rgb}{0,0.5,0}
\usepackage{hyperref}
\hypersetup{
    unicode=false,          
    colorlinks=true,        
    linkcolor=red,          
    citecolor=darkgreen,        
    filecolor=magenta,      
    urlcolor=cyan           
}
\usepackage[capitalize, nameinlink]{cleveref}
\crefname{theorem}{Theorem}{Theorems}
\crefname{proposition}{Proposition}{Propositions}
\crefname{observation}{Observation}{Observations}
\Crefname{lemma}{Lemma}{Lemmas}

\algnewcommand\algorithmicswitch{\textbf{switch}}
\algnewcommand\algorithmiccase{\textbf{case}}

\algdef{SE}[SWITCH]{Switch}{EndSwitch}[1]{\algorithmicswitch\ #1\ \algorithmicdo}{\algorithmicend\ \algorithmicswitch}%
\algdef{SE}[CASE]{Case}{EndCase}[1]{\algorithmiccase\ #1}{\algorithmicend\ \algorithmiccase}%
\algtext*{EndSwitch}%
\algtext*{EndCase}%

\newcommand{\eps}{\varepsilon}

\newcommand{\congest}{$\mathsf{CONGEST}$\xspace}
\newcommand{\local}{$\mathsf{LOCAL}$\xspace}

\newcommand{\poly}{\operatorname{\text{{\rm poly}}}}

\renewcommand{\paragraph}[1]{\vspace{0.15cm}\noindent {\bf #1}:}

\newcommand{\FullOrShort}{full}

\ifthenelse{\equal{\FullOrShort}{full}}{
  
  \newcommand{\fullOnly}[1]{#1}
  \newcommand{\shortOnly}[1]{}

  }{

    \newcommand{\fullOnly}[1]{}
    \newcommand{\IncludePictures}[1]{}
   
  }

\begin{document}
\title{Strong-Diameter Network Decomposition}
\author{
Yi-Jun Chang\thanks{Supported by  Dr.~Max R\"{o}ssler, by the Walter Haefner Foundation, and by the ETH Z\"{u}rich Foundation.}\\
\small ETH Z\"{u}rich \\
\and
Mohsen Ghaffari\thanks{Supported in part by a Starting Grant (grant agreement No.~853109) from the European Research Council (ERC), under the European Union’s Horizon 2020 research and innovation program.}\\
\small ETH Z\"{u}rich \\
}

\date{}
\maketitle
\thispagestyle{empty}
\setcounter{page}{0}

\begin{abstract}
Network decomposition is a central concept in the study of distributed graph algorithms. We present the first polylogarithmic-round deterministic distributed algorithm with small messages that constructs a strong-diameter network decomposition with polylogarithmic parameters.

Concretely, a ($C$, $D$) \emph{strong-diameter} network decomposition is a partitioning of the nodes of the graph into disjoint clusters, colored with $C$ colors, such that neighboring clusters have different colors and the subgraph induced by each cluster has a diameter at most $D$. In the \emph{weak-diameter} variant,  the requirement is relaxed by measuring the diameter of each cluster in the original graph, instead of the subgraph induced by the cluster. 

A recent breakthrough of Rozho\v{n} and Ghaffari [STOC 2020] presented the first $\text{poly}(\log n)$-round deterministic algorithm for constructing a weak-diameter network decomposition where $C$ and $D$ are both in $\text{poly}(\log n)$. Their algorithm uses small $O(\log n)$-bit messages. One can transform their algorithm to a strong-diameter network decomposition algorithm with similar parameters. However, that comes at the expense of requiring unbounded messages. The key remaining qualitative question in the study of network decompositions was whether one can achieve a similar result for strong-diameter network decompositions using small messages. We resolve this question by presenting a novel technique that can transform any black-box weak-diameter network decomposition algorithm to a strong-diameter one, using small messages and with only moderate loss in the parameters. 
\end{abstract}

\newpage
\section{Introduction}
Network decomposition is a central concept and a widely-used algorithmic tool in the area of distributed graph algorithms. In this paper, we present the first efficient (i.e., polylogarithmic-round) deterministic distributed algorithm with small messages that constructs a strong-diameter network decomposition with polylogarithmic parameters. This resolves one of the key remaining open problems in the study of network decompositions\footnote{This problem was 
stated as Open Problem 6 in Ghaffari's keynote at SIROCCO 2020~\cite{ghaffari2020network}.}.

\subsection{Background: Model and Definitions}
\paragraph{Model of Distributed Computing} We work with the standard \congest model of distributed computing~\cite{peleg00}. The network is abstracted as an $n$-node undirected unweighted graph $G=(V, E)$ where each node represents one processor in the network. We assume each node has a unique $O(\log n)$-bit identifier. Communication takes place in synchronous rounds, where per round each node can send one $B$-bit message to each of its neighbors --- typically, we assume $B=O(\log n)$. The relaxed variant where message sizes are not bounded is called the \local model~\cite{linial92}. Initially, nodes do not know the topology of the network $G$. At the end, each node should know its own part of the output, e.g., when computing a coloring of the graph, each node should know its own color, and when computing a clustering, each node should know its own cluster identifier (perhaps with some additional attributes of the cluster, such as its color and center node). The main measure of interest is the round complexity of the algorithm, i.e., the number of rounds until all nodes terminate and output. It is a common standard in the area to view $\poly(\log n)$ round complexity as a first-order interpretation of efficiency. That is, we would like to have distributed algorithms that run in $\poly(\log n)$ rounds, if not faster. This can be seen as an analog of viewing $\poly(n)$ time complexity as ``efficient'' in centralized computation.

\paragraph{Strong-diameter Network Decomposition} Network decomposition was introduced by Awerbuch, Goldberg, Luby, and Plotkin~\cite{awerbuch89}. Given an undirected graph $G=(V, E)$, a $C$-color $D$-diameter network decomposition --- sometimes also referred to as a $(C, D)$-network decomposition --- is to partition the set $V$ of all nodes into disjoint clusters $V_1, V_2, \dots$, meeting the following conditions.
\begin{itemize}
    \item Each cluster is assigned a color in $\{1, 2, \dots, C\}$ such that clusters that have neighboring nodes have different colors, i.e., any two clusters with the same color must be non-adjacent.
    \item The subgraph $G[V_i]$ induced by each cluster has diameter at most $D$, i.e., for any two nodes $u, v \in V_i$, there is a path of length at most $D$ that connects $u$ and $v$ and is made solely of nodes in $V_i$.
\end{itemize}


In an 
informal sense, network decomposition allows us to schedule various distributed computation tasks, so that we primarly have to deal with small-diameter clusters. The typical approach is to follow this template: we process the colors of the decomposition one by one. Per color, we process all clusters of this color at the same time. Since the clusters of one color are not adjacent, they can be processed simultaneously. Moreover, their small diameter facilities fast computation and coordination inside each cluster. To make it concrete, in this template, the time to process clusters of one color is proportional to the cluster diameter $D$. Since we have $C$ colors, the overall time is proportional to $C\cdot D$.   

\paragraph{Weak-diameter Network Decomposition} The network decomposition notion described above requires that the subgraph induced by each cluster has diameter at most $D$. That allows each cluster to perform its own communication and computation inside the induced subgraph of the cluster, and thus with no interference on the communications of the other cluster in the same color. This notion is sometimes referred to as a \emph{strong-diameter network decomposition} to distinguish it from a weaker variant, where we relax the second condition and allow the diameter to be measured in the original graph $G$. 

Formally, in a weak-diameter network decomposition, the diameter requirement is that for any two nodes $u, v \in V_i$, there is a path of length at most $D$ in graph $G$ that connects $u$ and $v$, but this path is allowed to include nodes that are not in the cluster $V_i$. In effect, the cluster is allowed to use some of the edges outside the cluster for its communication purposes. 

In the case of weak-diameter network decomposition, it is common to provide more structure which allows different clusters to simultaneously use these ``external" edges: it is required that each cluster $V_i$ has a Steiner tree $\mathcal{T}_i$ of depth at most $D$ in the original graph where all nodes of $V_i$ appear as terminal nodes of $\mathcal{T}_i$. Then, the network decomposition has a third parameter (besides the number of colors $C$ and diameter $D$) known as \emph{congestion} $L$: each edge $e\in G$ can appear in the Steiner trees of at most $L$ many clusters of the same color. Note that in the case of strong-diameter network decomposition, per color, each edge is used in the tree of at most one cluster of this color (the cluster that includes the two endpoints of this edge), and in this sense we have $L=1$.

\paragraph{Ball Carving and Relation to Network Decomposition} A concept closely related to network decomposition is that of ball carving, sometimes also referred to as low-diameter graph decomposition. In a \emph{strong-diameter} ball carving of diameter $D$, we receive a boundary parameter $\eps \in(0, 1)$ and then, we remove at most $\eps$ fraction of nodes and we cluster the remaining ones into non-adjacent clusters such that each cluster's induced subgraph has diameter at most $D$. We can also consider a relaxation, which we call \emph{weak-diameter} ball carving, where we require that each two nodes of the same cluster have distance at most $D$ in the original graph, instead of in the cluster's induced subgraph. Similarly, in the case of weak-diameter ball carving, we ask for an additional structure: each cluster  $V_i$ has a Steiner tree $\mathcal{T}_i$ of depth at most $D$ in the original graph where all nodes of $V_i$ appear as terminal nodes of $\mathcal{T}_i$. Again, we say the ball carving has congestion at most $L$ if each edge $e\in E$  appear in the Steiner trees of at most $L$ many clusters.

The relation to network decomposition is that we can obtain network decomposition via simple iterations of ball carving~\cite{linial93}: we repeat the ball carving process for $\log n$ iterations with boundary parameter $\eps=1/2$, each time on the nodes that remain not clustered in the previous iterations. Thus, each time we cluster at least half of the remaining nodes. Hence, within $\log n$ iterations, all nodes are clustered. Clusters that are defined in the $i$th iteration make up the clusters of the $i$th color of the network decomposition.

\paragraph{Applications of Network Decomposition and Ball Carving} Both network decomposition and ball carving are useful building blocks for designing distributed graph algorithms in the \congest model.
Ball carving was used~\cite{ChangS20,ChangS19} in designing distributed algorithms for  \emph{expander decomposition} and \emph{expander routing}, which are widely used tools with many applications~\cite{GhaffariKS17,GhaffariL2018,ChangS19,Daga2019distributed,EFFKO19,izumiLM20,ChatterjeePN20,SuVdsic19,CensorLL20,CensorCLL21}.  Ball carving was also recently used in distributed densest subgraph detection~\cite{su_et_al:LIPIcs:2020:13093}.
Deterministic  distributed algorithms for network decomposition can be applied to transform deterministic distributed algorithms with round complexity $O(D) \cdot \poly(\log n)$ into ones with round complexity $\poly(\log n)$. This approach was used in designing $\poly(\log n)$-round deterministic distributed algorithms for fundamental graph problems such as maximal independent set~\cite{censor2017derandomizing,RozhonG19} and $(\Delta+1)$-coloring~\cite{bamberger2020efficient} in the \congest model. Via other connections, these also led to improved randomized algorithms in the \congest model as well as the massively parallel computation model~\cite{chang2019coloring,ghaffari2016MIS,halldorsson2020efficient}.

\subsection{State of the Art} The pioneering work of Awerbuch, Goldberg, Luby, and Plotkin~\cite{awerbuch89} presented a deterministic distributed algorithm that computes a strong-diameter network decomposition with $2^{O(\sqrt{\log n\cdot \log\log n})}$ colors and diameter $2^{O(\sqrt{\log n\cdot \log\log n})}$ in  $2^{O(\sqrt{\log n\cdot \log\log n})}$ rounds of the \congest model\footnote{Awerbuch et al.~were not explicit about the message sizes, but $O(\log n)$ bit messages suffice for their algorithm.}. Shortly after, Panconesi and Srinivasan~\cite{panconesi-srinivasan} presented a modification of the approach of Awerbuch et al., which resulted in a deterministic distributed algorithm that computes a strong-diameter network decomposition with $2^{O(\sqrt{\log n})}$ colors and diameter $2^{O(\sqrt{\log n})}$ in  $2^{O(\sqrt{\log n})}$ rounds\footnote{We comment that this modification required using unbounded message sizes and for many years it was not know how to achieve these bounds with small messages. That was resolved recently in 2019: Ghaffari~\cite{ghaffari2019MIS} showed that a different modification of the approach of Awerbuch et al.~can directly obtain a network decomposition with $2^{O(\sqrt{\log n})}$ colors and diameter $2^{O(\sqrt{\log n})}$ in  $2^{O(\sqrt{\log n})}$ rounds, and using standard $O(\log n)$ bit messages.}. Linial and Saks~\cite{linial93} observed that for every $n$-node graph, there exists a strong-diameter network decomposition with $O(\log n)$ colors and $O(\log n)$ diameter; this was only an existential result and not a distributed algorithm. They also gave a randomized algorithm that computes a weak-diameter network decomposition with $O(\log n)$ colors and $O(\log n)$ diameter 
in $O(\log^2 n)$ rounds of the \congest model, with high probability\footnote{As standard, we use the phrase ``with high probability" (w.h.p.) to indicate that an even happens with probability at least $1-1/n^c$, for a desirably large constant $c\geq 2$.}. Much more recently, based on a technique developed by Miller, Peng, and Xu~\cite{miller2013parallel}, Elkin and Neiman~\cite{elkin16_decomp} strengthened this to a strong-diameter network decomposition: concretely, they presented a randomized distributed algorithm that computes a strong-diameter network decomposition with $O(\log n)$ colors and $O(\log n)$ diameter, in $O(\log^2 n)$ rounds of the \congest model, with high probability.    

In the context of network decomposition, we would like the number of colors $C$ and the diameter $D$ to be in $\poly(\log n)$. This allows us to use the decomposition following the standard template in $C\cdot D = \poly(\log n)$ rounds, hence resulting in an efficient distributed algorithm. In light of this, we refer to network decompositions with $\poly(\log n)$ colors and diameter as a network decomposition with \emph{good} parameters. 

A fundamental and long-standing question in the study of network decompositions was whether one can obtain an efficient deterministic algorithm for network decomposition with good parameters, having $\poly(\log n)$ colors and diameter, and working in $\poly(\log n)$ rounds. For comparison, a randomized counterpart was known due to the weak-diameter algorithm of Linial and Saks\cite{linial93} and the strong-diameter algorithm by Elkin and Neiman~\cite{elkin16_decomp}. Indeed, this question was at the center of the study of the gap between randomized and distributed graph algorithms; see \cite{ghaffari2017complexity, ghaffari2018derandomizing}. That question was finally resolved 
by Rozho\v{n} and Ghaffari~\cite{RozhonG19}, who presented a $\poly(\log n)$-round deterministic algorithm in the \congest model that computes a weak-diameter network decomposition with $\poly(\log n)$ colors, diameter, and congestion. If we allow unbounded messages, one can transform their algorithm to a $\poly(\log n)$-round \local model algorithm that computes a strong-diameter network decomposition with $O(\log n)$ colors and $O(\log n)$ diameter. This is based on a classic transformation algorithm of Awerbuch et al.~\cite{awerbuch96}.

However, it remained open --- perhaps as the last qualitative question in the study of network decompositions --- whether one can obtain an efficient deterministic distributed algorithm for strong-diameter network decomposition with good parameters and using small messages. 

\subsection{Our Contribution}
In this paper, we resolve the above question and present a $\poly(\log n)$-round deterministic algorithm in the \congest model that computes a strong-diameter network decomposition with good parameters. 

\begin{theorem}\label[theorem]{thm:main}
There are deterministic distributed algorithms in the \congest model that, in any $n$-node network $G=(V, E)$, compute a strong-diameter network decomposition of $G$ 
\begin{itemize}
    \item using $O(\log n)$ colors and with clusters of diameter $O(\log^3 n)$, in $O(\log^8 n)$ rounds, and
    \item using $O(\log n)$ colors and with clusters of diameter $O(\log^2 n)$, in $O(\log^{11} n)$ rounds.
\end{itemize}
\end{theorem}


\Cref{thm:main} is proved by combining \Cref{thm:main2} with  the standard reduction~\cite{linial93} from network decompositions to ball carving, i.e., with  $\log n$ repetitions of ball carving with $\eps = 1/2$.

\begin{theorem}\label[theorem]{thm:main2}
There are deterministic distributed algorithms in the \congest model that, in any $n$-node network $G=(V, E)$, compute a strong-diameter ball carving of $G$ 
\begin{itemize}
    \item  with clusters of diameter $O(\log^3 n /\eps)$, in $O(\log^7 n /\eps)$ rounds, and
    \item  with clusters of diameter $O(\log^2 n /\eps)$, in $O(\log^{10} n /\eps)$ rounds.
\end{itemize}
\end{theorem}


The key novelty in achieving \Cref{thm:main2} is a 
message-efficient
deterministic reduction that can efficiently transform \emph{any} 
algorithm for weak-diameter ball carving into an algorithm for strong-diameter ball carving, with only moderate loss in the diameter and round complexity. If the former algorithm is deterministic, so is the latter. 

Such a transformation was previously known using unbounded message sizes, due to a technique of Awerbuch et al.\cite{awerbuch96}. However, that technique heavily relies on these large messages, and uses them to gather the entire topology around certain clusters to perform computation in a centralized fashion. Our transformation technique is completely different and works with small messages. 
Such a transformation was not previously known even when allowing randomized algorithms. Indeed, while an efficient randomized algorithm for weak-diameter network decomposition was known since the  work of Linial and Saks~\cite{linial93} in 1993, it took until 2016 that Elkin and Neiman~\cite{elkin16_decomp} presented an efficient randomized algorithm for strong-diameter network decomposition.  Their approach gives a new strong-diameter network decomposition algorithm and not a general transformation from weak-diameter network decompositions.   

See \cref{tab:decomposition,tab:carving} for a list of our results about network decomposition and ball carving, comparing with the results from previous work.
 We emphasize that all results in \cref{tab:carving} not only apply to the \emph{node version} of ball carving, but they also apply to the  \emph{edge version}, where  we remove at most an $\eps$ fraction of the edges, instead of removing nodes. The proofs for the edge version are essentially the same as that for the node version, so they are omitted for simplicity.
 
 It  remains an intriguing open question whether a strong-diameter ball carving with diameter $O(\log n / \eps)$ and a strong-diameter network decomposition with $O(\log n)$ colors and $O(\log n)$ diameter can be constructed in polylogarithmic rounds in the \congest model.

\begin{table*}[ht]
    \centering
        \caption{Network decomposition in the \congest model}
    \begin{tabular}{|l|l|l|l|l|l|}
\multicolumn{1}{l}{\bf Type} & 
\multicolumn{1}{l}{\bf Model} & 
\multicolumn{1}{l}{\bf Colors} & 
\multicolumn{1}{l}{\bf Diameter} &
\multicolumn{1}{l}{\bf Rounds} &
\multicolumn{1}{l}{\bf Reference}
\\ \hline
\multirow{3}{*}{Weak} & Randomized & $O(\log n)$ & $O(\log n)$ &  $O(\log^2 n)$ & \cite{linial93}
\\ \cline{2-6}
& \multirow{2}{*}{Deterministic} & $O(\log n)$ & $O(\log^3 n)$ & $O(\log^7 n)$ & \cite{RozhonG19}
\\ \cline{3-6}
&  & $O(\log n)$ & $O(\log^2 n)$ & $O(\log^5 n)$ & \cite{ghaffariGR2021improvedND}
\\ \hline
\multirow{5}{*}{Strong} & Randomized & $O(\log n)$ & $O(\log n)$ &  $O(\log^2 n)$ & \cite{miller2013parallel,elkin16_decomp}
\\ \cline{2-6}
& \multirow{4}{*}{Deterministic} & $2^{O(\sqrt{\log n \log \log n})}$ & $2^{O(\sqrt{\log n \log \log n})}$ & $2^{O(\sqrt{\log n \log \log n})}$ & \cite{awerbuch89}
\\ \cline{3-6}
&  & $2^{O(\sqrt{\log n})}$ & $2^{O(\sqrt{\log n})}$ & $2^{O(\sqrt{\log n})}$ & \cite{ghaffari2019MIS,panconesi-srinivasan} 
\\ \cline{3-6}
&  & $O(\log n)$ & $O(\log^3 n)$ & $O(\log^8 n)$ & \cref{thm:decomposition}
\\ \cline{3-6}
&  & $O(\log n)$ & $O(\log^2 n)$ & $O(\log^{11} n)$ & \cref{thm:decomposition2}
\\ \hline
    \end{tabular}
    \label{tab:decomposition}
\end{table*}

\begin{table*}[ht]
    \centering
        \caption{Ball carving in the \congest model}
    \begin{tabular}{|l|l|l|l|l|}
\multicolumn{1}{l}{\bf Type} & 
\multicolumn{1}{l}{\bf Model} & 
\multicolumn{1}{l}{\bf Diameter} &
\multicolumn{1}{l}{\bf Rounds} &
\multicolumn{1}{l}{\bf Reference}
\\ \hline
\multirow{3}{*}{Weak} & Randomized  & $O(\log n  / \eps)$ &  $O(\log n  / \eps)$ & \cite{linial93}
\\ \cline{2-5}
& \multirow{2}{*}{Deterministic}  &  $O(\log^3 n  / \eps)$ & $O(\log^6 n  / \eps^2)$ & \cite{RozhonG19}
\\ \cline{3-5}
& & $O(\log^2 n  / \eps)$  & $O(\log^4 n  / \eps^2)$ & \cite{ghaffariGR2021improvedND}
\\ \hline
\multirow{3}{*}{Strong} & Randomized  & $O(\log n / \eps)$ &  $O(\log n  / \eps)$ & \cite{miller2013parallel,elkin16_decomp}
\\ \cline{2-5}
& \multirow{2}{*}{Deterministic}   & $O(\log^3 n  / \eps)$  & $O(\log^7 n  / \eps^2)$& \cref{thm:carving}
\\ \cline{3-5}
& & $O(\log^{2} n  / \eps)$  & $O(\log^{10} n  / \eps^2)$ & \cref{thm:carving2}
\\ \hline
    \end{tabular}
    \label{tab:carving}
\end{table*}

\subsection{Our Method
}
The core technical novelty in our work is presenting a transformation algorithm that converts an algorithm for weak-diameter ball carving into an algorithm for strong-diameter ball carving. 
We first provide a brief overview of the transformation technique of Awerbuch et al., and then we outline our transformation technique.

\paragraph{Recap on the Transformation Technique of Awerbuch et al.~\cite{awerbuch96}} Suppose that we have a weak-diameter network decomposition algorithm $\mathcal{A}$ that on $n$-node networks finds clusters with $C(n)$ colors and $D(n)$ weak diameter. We describe how Awerbuch et al.~obtain a strong-diameter ball carving with boundary parameter $\eps=1/2$. The overall strong-diameter network decomposition then follows by the standard connection from ball carving to network decomposition as mentioned above. 

Awerbuch et al.~first run this decomposition algorithm $\mathcal{A}$ on $G^{2d}$, where $d=\log n$. Here, $G^{2d}$ denotes the power graph where we put an edge between any two nodes whose distance is at most $2d$. Thus, we get a weak-diameter network decomposition with $C(n)$ colors and clusters of weak-diameter $D(n) \cdot 2d = D(n) \cdot \Theta(\log n)$ in $G$, where any two clusters of the same color have distance at least $2d+1$. Then, we process the colors of this weak-diameter network decomposition one by one, following the standard template of using network decomposition: per color, each cluster $\mathcal{C}$ gathers into the center of the cluster the topology of the entire cluster as well as $d$-hop neighborhood of the cluster. Note that since any two clusters have distance at least $2d+1$, these gathered topologies are disjoint. Then, the cluster $\mathcal{C}$ simulates a sequential ball carving process in a centralized fashion to define the output strong-diameter balls: each time, we pick a node $v \in \mathcal{C}$ that remains not clustered in the output ball carving, and we find the smallest value $r$ such that $|B_{r+1}(v)|/|B_{r}(v)| \leq 2$. Here, $B_{r}(v)$ denotes the set of all nodes within distance at most $r$ from node $v$, in the subgraph induced by the remaining node. It is easy to see that $r \leq \log n = d$. We then take $B_{r}(v)$ as one strong-diameter cluster of the output ball carving, and we remove all of its nodes. We also remove all nodes of $B_{r+1}(v)\setminus B_{r}(v)$ and declare them dead, meaning that they are not clustered in our strong-diameter ball carving. Since $|B_{r+1}(v)|/|B_{r}(v)| \leq 2$, we clustered at least half of the nodes that we removed. We then proceed to the next node $v'$ in the cluster $\mathcal{C}$, if any node remains. This is a very sequential process. However, we are performing it in a centralized manner as a local computation at the center of the cluster. Once the cluster center computes all these strong-diameter clusters, it informs all the nodes. These define the clusters of our desired ball carving. Note that computations of two different weak-diameter clusters never interfere as each cluster works only within $d = \log n$ neighborhood of its nodes, and each two clusters have distance at least $2d+1 = 2\log n+1$. Once we process all colors of the weak-diameter network decomposition, we have defined the clusters of the output ball carving, and at least half of the nodes have been clustered. As mentioned before, this strong-diameter ball carving can then be used easily, by $d = \log n$ repetitions, to construct the desired strong-diameter network decomposition. 

\paragraph{Our Transformation} We outline our transformation here for the case $\eps = 1/2$ while ignoring several details and hiding some lower-ordering technicalities. The actual transformation is explained in the next section. Suppose that we have a weak-diameter ball carving algorithm $\mathcal{A}$ that, given a parameter $\eps'$, on any $n$-node network removes at most $\eps'$ fraction of nodes and clusters the remaining ones into non-adjacent clusters with $R(n, \eps)$-depth Steiner trees. 
In our particular case, we will apply the weak-diameter ball carving algorithm 
of Ghaffari, Grunau, and Rozho\v{n}~\cite{ghaffariGR2021improvedND}, which is an optimized variant of the 
one of Rozho\v{n} and Ghaffari~\cite{RozhonG19}. 

Our algorithm has $\log n$ iterations. In the first iteration, we apply algorithm $\mathcal{A}$ with parameter $\eps'=1/(4\log n)$ and obtain some weak-diameter clusters. There are two possibilities: 
\begin{itemize}
    \item[(I)] Suppose that one of these clusters is \emph{giant}, meaning that it includes at least half of the nodes. Then, we perform a ball carving just on this one cluster, starting with the root node of the Steiner tree for the cluster and with an initial radius that ensures to cover all nodes of this giant cluster. We then gradually grow the radius hop by hop until we find a place where the boundary is of size at most $1/4$ of the nodes inside the ball. Such a radius is found within $O(\log n)$ steps of growth, as otherwise the ball size would go beyond $n$. From a computational perspective, it suffices to gather the sizes of the BFS layers around the chosen node of the giant cluster, and this can be easily performed in the \congest model. Thus, we get a strong-diameter ball with diameter $2R(n, 1/(4\log n)) + O(\log n)$ that covers the entire giant cluster. We then remove this found ball as one cluster of the output strong-diameter ball carving. Since we covered all of the giant cluster, we have clustered at least half of the nodes. We then proceed to the next iteration in the remaining nodes. 
    
    \item[(II)] In the second case, there is no such giant cluster and each cluster has at most half of the nodes. In this case, each connected component of the subgraph induced by alive nodes is also made of at most half of the nodes, because each connected component is a subset of one cluster (otherwise, there would be two adjacent clusters in one connected component). In this case, we simply move to the next iteration and we process each connected component separately. 
\end{itemize}
In either of the two cases (I) and (II), the connected component has at most half of the number of nodes in the previous iteration. The next iteration then repeats the process on each connected component separately. Since each time the size of the connected component shrinks by a  factor of $2$, we are done in $\log n$ iterations. Moreover, per iteration, the weak-diameter carving removes at most $\eps'=1/(4\log n)$ fraction of nodes. Hence, even over all the iterations, we remove at most $1/4$ of the nodes in the course of applying weak-diameter carving. On the other hand, the strong-diameter  carving algorithm that we apply in case (I) clusters some number $k$ of nodes, as one cluster of the output strong-diameter ball carving, and removes at most $k/4$ nodes as boundary. Since each node appears at most once as a part of the output strong-diameter ball carving, overall we lose at most $1/4$ of the nodes because of these strong-diameter network decompositions. Hence, taking both kinds of node removals into account, we remove at most $\eps = 1/2$ fraction of nodes.   
We present the formal and detailed description of this transformation in \Cref{sec:transformation}. 

\paragraph{Improving Diameter} The above transformation procedure loses a small $O(\log n)$ factor in the diameter of the clusters, because we have to set the boundary parameter of the employed weak-diameter carving smaller by a factor $O(\log n)$ to allow room for the boundary removals of different iterations. To mitigate this, in \Cref{sec:diameter-imp}, we explain a different algorithm that processes these larger strong-diameter balls and improves their diameter. Hence the overall transformation gets us to balls of strong-diameter $O(\log^2 n)$, at the expense of only $O(\log^3 n)$ factor larger number of rounds in the \congest model. 





\section{Strong-diameter Ball Carving via Weak-diameter Ball Carving}
\label[section]{sec:transformation}
In this section we describe our \textit{strong-diameter ball carving} algorithm.
For the sake of simplicity, in the proof of \cref{thm:transformation} we assume that  the number of nodes $n$ is a global knowledge. This assumption can be removed  by first constructing a weak-diameter ball carving with boundary parameter $\eps/2$,  counting the number of nodes $n' = |\mathcal{C}|$ in each cluster $\mathcal{C}$, applying the algorithm of \cref{thm:transformation} with boundary parameter $\eps/2$ to the subgraph induced by each cluster $\mathcal{C}$ in parallel. 
Alternatively, the assumption that each node has a unique identifiers of length $\ell = \Theta(\log n)$ implies that $2^\ell$ is an upper bound on the number of nodes, and  we may also use $n = 2^\ell$ in \cref{thm:transformation}.


 
In the description of \cref{thm:transformation}, we note that a cluster $\mathcal{C}$ resulting from the given weak-diameter ball carving algorithm $\mathcal{A}$ can have unbounded strong diameter, and it may even induce a disconnected subgraph. However, $\mathcal{C}$  has weak diameter at most $2R(n, \eps)$ due to the Steiner tree $\mathcal{T}$, which may contain nodes outside of the cluster $\mathcal{C}$.

\begin{theorem}\label[theorem]{thm:transformation}
Suppose that there is a distributed algorithm $\mathcal{A}$ in the \congest model that, on any $n$-node graph and given a parameter $\eps$, removes at most an $\eps$ fraction of the nodes and clusters the remaining ones into non-adjacent clusters meeting the following conditions:
\begin{itemize}
    \item Each cluster $\mathcal{C}$ has a Steiner tree $\mathcal{T}$ such that all nodes in $\mathcal{C}$ appear as terminals of $\mathcal{T}$. 
    \begin{itemize}
    \item The depth of $\mathcal{T}$ is at most $R(n, \eps)$, i.e., the distance between the root node and any terminal node in the tree is at most  $R(n, \eps)$.
    \item Each edge belongs to at most $L(n, \eps)$ Steiner trees.
    \end{itemize}
    \item The algorithm takes at most  $T(n, \eps)$ rounds.
\end{itemize}
Then, there is a distributed algorithms $\mathcal{B}$ in the \congest model that, on any $n$-node graph and given a parameter $\eps$, removes at most an $\eps$ fraction of the nodes so that each remaining connected component has strong diameter $2R(n, \eps/(2\log n)) + O(\log n/\eps)$, with round complexity 
\ifdefined\DoubleColumn
\begin{align*}
& O(\log n) \cdot T(n, \eps/(2\log n)) \\
& +
O(\log n) \cdot  R(n, \eps/(2\log n)) \cdot L(n, \eps/(2\log n))\\
& +
O(\log^2 n / \eps).
\end{align*}
\else
\[
O(\log n) \cdot T(n, \eps/(2\log n)) +
O(\log n) \cdot  R(n, \eps/(2\log n)) \cdot L(n, \eps/(2\log n)) +
O(\log^2 n / \eps).
\]
\fi
\end{theorem}
\begin{proof}
We  build the claimed strong-diameter ball carving algorithm $\mathcal{B}$ via black-box invocations to the promised weak-diameter ball carving algorithm $\mathcal{A}$. The algorithm is made of a number of iterations. We next first present the outline of what these iterations aim to achieve, and then explain how each iteration works. We conclude then by analyzing this algorithm and arguing that it proves the theorem.

\paragraph{Algorithm Outline}
The algorithm has $\log n$ iterations. During these iterations, we define some strong-diameter balls and remove the nodes on the boundary and we also declare some more nodes dead. As a result, the algorithm zooms into smaller and smaller connected components of the subgraph induced by alive nodes. The algorithm guarantees that 
at the start of the $i$th iteration, each connected component of alive nodes has size at most $n/2^{i-1}$.  
Moreover, we declare at most $\eps$ fraction of the nodes \emph{dead}, throughout the algorithm. These dead nodes are not included in the final clustering.
We emphasize that different copies of the algorithm run in different components independently and simultaneously. 

\paragraph{One Iteration} Consider iteration $i$. Suppose that we are now considering one connected component of alive nodes and the set of nodes in this component is $S$. By the guarantee from the previous iterations, we have $|S|\leq n/2^{i-1}$. 

We first run the weak-diameter ball carving algorithm $\mathcal{A}$ on $G[S]$ with boundary parameter $\eps'= \eps/ (2\log n)$. As a result, we get clusters in $G[S]$ that are non-adjacent, 
while at most $\eps'$ fraction of nodes of $S$ are not clustered.  The clusters are augmented with $R(n, \eps/(2\log n))$-depth Steiner trees in graph $G[S]$, one for each cluster, with congestion $L(n, \eps/(2\log n))$. 

Considering these clusters, we have two cases: 
\begin{itemize}
    \item[(I)] First, suppose that each of the clusters in this weak-diameter ball carving has at most $n/2^{i}$ nodes. In this case,
     we declare those nodes that are not clustered \emph{dead} for the rest of the strong-diameter ball carving process, and then we move to the next iteration, where each of the connected components of the alive nodes is handled separately. Note that if each cluster has at most $n/2^{i}$ nodes, then definitely each connected component of the alive nodes in $S$ also has at most $n/2^{i}$ nodes, as each connected component is a subset of one cluster.
    
    \item[(II)] Second, suppose that there is one cluster $\mathcal{C}$ such that $|\mathcal{C}|> n/2^{i}$. Note that we can have at most one such cluster because $|S|\leq n/2^{i-1}$. Let $a$ be the root of the Steiner tree $\mathcal{T}$  for $\mathcal{C}$. We run a sequential ball carving process from $a$ starting with radius $R(n, \eps/(2\log n))$, and for $O( \log n/ \eps)$ radius growth steps\footnote{The diameter of  $G[S]$ may be smaller than $R(n, \eps/(2\log n))$, in which case the BFS from $a$ covers all nodes in $S$  in less than $R(n, \eps/(2\log n))$ steps.  In this case, we will have  $r^\ast = R(n, \eps/(2\log n))$,  $B_{r^\ast}(a) = B_{r^\ast + 1}(a)  = S$, and $S$ will be a cluster of the final strong-diameter ball carving.}, in the whole component $G[S]$. 
    
    Concretely, we find the smallest value $r^*$ for the radius parameter \[r\in [R(n, \eps/(2\log n)), R(n, \eps/(2\log n))+O( \log n / \eps)]\] such that we have $|B_{r}(a)|/|B_{r+1}(a)|\geq 1 - \eps/2$.
    Here, $B_{r}(a)$ denotes all nodes in $G[S]$ that have distance at most $r$ from node $a$, where distance is measured in the graph $G[S]$. Such a value exists as we cannot have more than $O( \log n / \eps)$ steps of growth by a factor of $1/(1-\eps/2)$, since otherwise the number of nodes in $S$ exceeds $n$, which is impossible. Moreover,  since the depth of $\mathcal{T}$ is at most $R(n, \eps / (2 \log n))$, we know that $\mathcal{C}\subseteq B_{R(n, \eps/(2\log n))}(a)$. 
    
    We can compute the value of $r^*$ by performing a simple BFS from $a$ in $G[S]$ and then gathering at $a$ the number of nodes within each distance \[r \in [R(n, \eps/(2\log n)), R(n, \eps/(2\log n)) +O(\log n/\eps)].\] Once $r^*$ is found, we declare $B_{r^*}(a)$ as one cluster of the final strong-diameter ball carving and put it aside, and we declare nodes of $B_{r^*+1}(a)\setminus B_{r^*}(a)$ as dead. Then, we remove nodes of $B_{r^*+1}(a)$ from $S$.  We then proceed to the next iteration, where we handle each connected component of the alive nodes in $S$ separately. Because $|S| \leq n/2^{i-1}$ and $|B_{r^*+1}(a)| \geq |\mathcal{C}|> n/2^{i}$, each connected component of the alive nodes in $S$ has at most $n/2^{i}$ nodes.
\end{itemize}

\paragraph{Correctness} At the end of  iteration $i$, the size of each connected component of alive nodes is at most $n/2^{i}$. Hence, after $\log n$ iterations, each remaining connected component is trivial and makes its own cluster in the final strong-diameter ball carving output. 

The fraction of nodes that are dead due to the ball carving algorithm $\mathcal{A}$ is at most $\eps' \log n = \eps/2$.
The fraction of nodes that are dead because of the  Case~(II) above is also at most $\eps/2$, as we can charge the cost of the dead nodes $B_{r^*+1}(a)\setminus B_{r^*}(a)$  to the ball $B_{r^*+1}(a)$, and $|B_{r^*+1}(a)\setminus B_{r^*}(a)| \leq (\eps/2)|B_{r^*+1}(a)|$. Hence the   dead nodes constitute of at most an $\eps$ fraction of all nodes.

 The strong-diameter of each cluster in the final clustering is at most $2R(n, \eps/(2\log n)) + O(\log n/\eps)$, because we have $r^\ast = R(n, \eps/(2\log n)) + O(\log n/\eps)$ in the algorithm.

\paragraph{Round Complexity} Now we analyze the round complexity of one iteration of the algorithm. 
\begin{itemize}
    \item  The invocation of $\mathcal{A}$ costs $T(n, \eps/(2\log n))$ rounds.
    \item  Checking whether there is one cluster $\mathcal{C}$ such that $|\mathcal{C}|> n/2^{i}$ in $G[S]$ can be done by an information gathering using the Steiner trees. This costs $R(n, \eps/(2\log n)) \cdot L(n, \eps/(2\log n))$ rounds.
    \item The computation of $r^\ast$ via BFS takes 
\ifdefined\DoubleColumn
\[O(r^\ast) = O(R(n, \eps/(2\log n)) + \log n/\eps)\]
\else $O(r^\ast) = O(R(n, \eps/(2\log n)) + \log n/\eps)$
\fi rounds.
\end{itemize} 
Hence the overall round complexity is
\ifdefined\DoubleColumn
\begin{align*}
& O(\log n) \cdot T(n, \eps/(2\log n)) \\
& +
O(\log n) \cdot  R(n, \eps/(2\log n)) \cdot L(n, \eps/(2\log n)) \\
& +
O(\log^2 n / \eps).\qedhere
\end{align*}
\else
\[
O(\log n) \cdot T(n, \eps/(2\log n)) +
O(\log n) \cdot  R(n, \eps/(2\log n)) \cdot L(n, \eps/(2\log n)) +
O(\log^2 n / \eps).\qedhere
\]
\fi
\end{proof}

Combining \Cref{thm:transformation} with the work of Ghaffari, Grunau, and Rozho\v{n}~\cite{ghaffariGR2021improvedND}, we obtain the following results.

\begin{theorem}\label[theorem]{thm:carving}
There is an $O(\log^7 n/\eps^2)$-round deterministic distributed algorithm in the \congest model that computes a strong-diameter ball carving  of an $n$-node graph $G$ that removes an $\eps$ fraction of the nodes so that each remaining connected component has strong diameter  $D = O(\log^3 n/\eps)$.
\end{theorem}
\begin{proof}
The work of Ghaffari, Grunau, and Rozho\v{n}~\cite{ghaffariGR2021improvedND} provides an algorithm $\mathcal{A}$ that 
on any $n$-node graph and given a parameter $\eps$, removes at most an $\eps$ fraction of the nodes and clusters the remaining ones into non-adjacent clusters meeting the following conditions:
\begin{itemize}
    \item Each cluster $\mathcal{C}$ has a Steiner tree $\mathcal{T}$ where all nodes in $\mathcal{C}$ appear as terminals of $\mathcal{T}$. 
    \begin{itemize}
    \item The diameter of $\mathcal{T}$ is at most $R(n, \eps) = O(\log^2 n/\eps)$.
    \item Each edge belongs to at most $L(n, \eps) = O(\log n)$ Steiner trees.
    \end{itemize}
    \item The algorithm takes at most  $T(n, \eps) = O(\log^4 n/\eps^2)$ rounds.
\end{itemize}
By plugging this directly into \Cref{thm:transformation}, we get an algorithm $\mathcal{B}$ that on any $n$-node graph and given a parameter $\eps$, removes at most an $\eps$ fraction of the nodes so that each remaining connected component has strong diameter $2R(n, \eps/(2\log n)) + O(\log n/\eps) = O(\log^3 n/\eps)$. This algorithm's round complexity is 
\ifdefined\DoubleColumn
\begin{align*}
& O(\log n) \cdot T(n, \eps/(2\log n)) \\
& \ \  +
O(\log n) \cdot  R(n, \eps/(2\log n)) \cdot L(n, \eps/(2\log n)) \\
& \ \ +
O(\log^2 n / \eps) \\  & = O(\log^7 n/\eps^2) + O(\log^5 n/\eps) + O(\log^2 n/\eps)  \\
  & = O(\log^7 n/\eps^2). \qedhere
\end{align*}
\else
\begin{align*}
& O(\log n) \cdot T(n, \eps/(2\log n)) +
O(\log n) \cdot  R(n, \eps/(2\log n)) \cdot L(n, \eps/(2\log n)) +
O(\log^2 n / \eps) \\  & = O(\log^7 n/\eps^2) + O(\log^5 n/\eps) + O(\log^2 n/\eps)  \\
  & = O(\log^7 n/\eps^2). \qedhere
\end{align*}
\fi
\end{proof}

By the standard reduction from network decompositions to ball carving, we obtain the following theorem.

\begin{theorem}\label[theorem]{thm:decomposition}
There is an $O(\log^8 n)$-round deterministic distributed algorithm in the \congest model that computes a strong-diameter network decomposition  of an $n$-node graph $G$ using $C = O(\log n)$ colors and with clusters of diameter $D = O(\log^3 n)$.
\end{theorem}
\begin{proof}
We obtain the desired strong-diameter network decomposition via repeated application of the ball carving algorithm of \cref{thm:carving} with  $\eps=1/2$. In each application, we cluster at least half of the nodes into non-adjacent clusters of strong-diameter $O(\log^3 n)$, in $O(\log^7 n)$ rounds. We then remove all these clustered nodes and repeat on the remaining nodes. In each application, half of the remaining nodes get clustered; those clustered in the $i$th iteration receive color $i$. Hence, after $O(\log n)$ iterations, all nodes are clustered. Thus, we obtain our network decomposition with $O(\log n)$ colors, $O(\log^3 n)$ strong-diameter, and using $O(\log^8 n)$ rounds.
\end{proof}

\section{Improving Diameter}
\label[section]{sec:diameter-imp}

The  transformation procedure in \cref{sec:transformation} loses an  $O(\log n)$ factor in the diameter of the clusters. To mitigate this, in this section we explain a different algorithm that processes these larger strong-diameter balls and improves their diameter to $O(\log^2 n / \eps)$. We begin with a key lemma.

\begin{lemma}\label[lemma]{lem:tool}
Given  $0 < \eps < 1$, there is a distributed algorithm in the \congest model  on any $n$-node $D$-diameter graph $G=(V,E)$ returns either one of the following in $O(D \log n)$ rounds.
\begin{itemize}
    \item {\bf Balanced Sparse Cut:} Find two non-adjacent node sets $V_1 \subseteq V$ and $V_2 \subseteq V$ such that $|V_1| \geq n/3$, $|V_2| \geq n/3$, and $|V \setminus (V_1 \cup V_2)|  = O(\eps n / \log n)$. 
    \item {\bf Large  Small-diameter Component:} Find an $O(\log^2 n / \eps)$-diameter component $U \subseteq V$ such that $|U| \geq n/3$ and the number of nodes in $V \setminus U$ adjacent to $U$ is $O(\eps n / \log n)$.
\end{itemize}
\end{lemma}
\begin{proof}
The algorithm has $O(\log n)$ iterations, and each iteration costs $O(D)$ rounds. We maintain a set $S$ throughout the algorithm, and initially we set $S = V$ to be the set of all nodes.
For any $S \subseteq V$, we consider the following two parameters:
\begin{itemize}
    \item $a$ is the smallest number such that the radius-$a$ neighborhood of $S$ has at least $n/3$ nodes.
     \item $b$ is the smallest number such that the radius-$b$ neighborhood of $S$ has at least $2n/3$ nodes.
\end{itemize}

The design of algorithm aims to satisfy the following induction hypothesis for the set $S$ at the start of the $i$th iteration:
\begin{itemize}
    \item $|S| \leq n \cdot 2^{-(i-1)}$.
    \item $a = O((i-1) \log n / \eps)$.
\end{itemize}

Initially, we have $|S| = n$ and $a = 0$.
After $O(\log n)$ iterations, we end up with $S = \{v\}$ for some node $v$  with $a = O(\log^2 n / \eps)$.
After that, we start a  BFS  from $v$ to find the  layer number $r^\ast$ that minimizes $|B_{r+1}(v)|/|B_{r}(v)|$ among all layer numbers $a \leq r \leq a + O(\log n / \eps)$.
It is guaranteed that $|B_{r^\ast+1}(v)|/|B_{r^\ast}(v)| = 1 + O(\eps / \log n)$, and so $B_{r^\ast+1}(v) \setminus B_{r^\ast}(v)$ has at most $O(\eps n / \log n)$ nodes.
Setting $U \leftarrow B_{r^\ast}(v)$  gives us the required $O( \log^2 n  / \eps)$-diameter subgraph with size at least $n/3$.

\paragraph{One Iteration}
Consider the set $S$ at the beginning of this iteration.
If $b - a = \Omega(\log n / \eps)$, then we are in the good case in that we can find a balanced sparse cut by cutting along the weakest layer among these $b-a$ layers. More concretely, let $B_{k}(S)$ denote the set of all nodes within distance $k$ to $S$. We find the  layer number $r^\ast$ that minimizes $|B_{r+1}(S)|/|B_{r}(S)|$ among all layer numbers $a \leq r \leq b-2$.
It is guaranteed that $|B_{r^\ast+1}(S)|/|B_{r^\ast}(S)| = 1 + O(\eps / \log n)$, and so $B_{r^\ast+1}(S) \setminus B_{r^\ast}(S)$ has at most $O(\eps n / \log n)$ nodes.
Setting $V_1 \leftarrow B_{r^\ast}(S)$ and $V_2 \leftarrow V \setminus B_{r^\ast+1}(S)$ gives us the required balanced sparse cut, and the algorithm terminates. 

In what follows, we focus on the case of $b - a = O(\log n / \eps)$.
We divide the set $S$ into two equal-sized sets $S_1$ and $S_2$ arbitrarily. This can be done in $O(D)$ rounds as follows. Identify the node $v^\ast$ with the smallest identifier in the graph. Find a BFS tree starting from $v^\ast$. Sort the nodes in $S$ according to the in-order traversal of the tree. Set $S_1$ to be the first half of the nodes in the sorted order, and set $S_2 = S \setminus S_1$.

We consider the four parameters $a_1
$, $b_1$, $a_2$, and $b_2$ associated with $S_1$ and $S_2$. A crucial observation is the following:
\[\min \{a_1, a_2\} \leq b = a + (b-a) = a + O(\log n / \eps).\]
To see this inequality, let $W$ (resp., $W_1$ and $W_2$) be the set of all nodes within distance $b$ to $S$ (resp., $S_1$ and $S_2$). We have $W = W_1 \cup W_2$, so $\max\{|W_1|, |W_2|\} \geq |W|/2 \geq n/3$.
Note that $|W_i| \geq n/3$ implies that $a_i \leq b$, and  hence  
$\min \{a_1, a_2\} \leq b$.

If $a_1 < a_2$, we move on to the next iteration with $S \leftarrow S_1$.
Otherwise, we   move on to the next iteration with $S \leftarrow S_2$.
In either case, the induction hypothesis is met, as  $\min \{a_1, a_2\} \leq a +  O(\log n / \eps)$.

The cost of one iteration is $O(D)$ rounds. 
The calculation of the parameters $a$ and $b$, the task of finding a balanced sparse cut when  $b - a = \Omega(\log n / \eps)$, and the partition of $S$ into $S_1$ and $S_2$ can all be done in $O(D)$ rounds. 
\end{proof}

Using \cref{lem:tool}, we prove the following theorem, which transforms \emph{any} strong-diameter ball carving algorithm into a strong-diameter ball carving algorithm with diameter bound $O(\log^2 n /\eps)$.

\begin{theorem}\label[theorem]{thm:transformation2}
Suppose that there is a distributed algorithm $\mathcal{A}$ in the \congest model that, on any $n$-node graph and given a parameter $\eps$, removes at most an $\eps$ fraction of the nodes so that each remaining connected component has strong diameter $R(n, \eps)$, and using $T(n, \eps)$ rounds.  
Then, there is a distributed algorithms $\mathcal{B}$ in the \congest model that attains the strong diameter bound $O(\log^2 n /\eps)$, with a round complexity of 
\[
O(\log n) \cdot T(n, \Theta(\eps / \log n)) +
O(\log^2 n) \cdot R(n, \Theta(\eps / \log n)). 
\]
\end{theorem}
\begin{proof}
Similarly, we build the strong-diameter ball carving algorithm  $\mathcal{B}$  via black-box invocations to the strong-diameter ball carving algorithm $\mathcal{A}$. 

\paragraph{Algorithm}
Informally, the algorithm  applies the algorithm of \cref{lem:tool} recursively to each strong-diameter cluster in a given strong-diameter ball carving. This will turn the given strong-diameter ball carving into one with an $O(\log^2 n / \eps)$ diameter bound, at the cost of removing some nodes. After each step of recursion, the number of nodes in each part will be reduced by a constant factor, so there is at most $O(\log n)$ levels of recursion. Moreover, at the beginning of each level of recursion, we will have to  run a strong-diameter ball carving algorithm again because the  diameter of the subgraph currently under consideration might be unbounded.

More concretely, let $\mathcal{A}_1$ be the  strong-diameter ball carving algorithm $\mathcal{A}$ promised in the theorem, and let $\mathcal{A}_2$ be the algorithm for \cref{lem:tool}. The algorithm $\mathcal{B}$ is as follows.

\begin{itemize}
    \item Run $\mathcal{A}_1$ with parameter $\eps' = \Theta(\eps / \log n)$. All nodes not in a cluster are dead.
    \item For each cluster $\mathcal{C}$, run  $\mathcal{A}_2$ on the subgraph induced by $\mathcal{C}$. Based on the outcome of $\mathcal{A}_2$, there are two cases.
    \begin{itemize}
        \item If the outcome is a balanced sparse cut, then we recurse on both of $G[V_1]$ and $G[V_2]$, in parallel. All nodes in $\mathcal{C} \setminus (V_1 \cup V_2)$ are dead.
        \item If the outcome is a large small-diameter component, then we add $U$ to be a cluster in the final clustering, and recurse on the subgraph induced by the nodes in $\mathcal{C}$ that are not adjacent to $U$. All nodes in $\mathcal{C} \setminus U$ that are adjacent to $U$ are dead.
    \end{itemize}
\end{itemize}

\paragraph{Analysis} In each iteration, at most $\eps' = O(\eps / \log n)$ fraction of the nodes are dead due to the ball carving algorithm $\mathcal{A}_1$, and at most $O(\eps / \log n)$ of the nodes are dead due to the post-processing for the outcome  of $\mathcal{A}_2$.  Since there are $O(\log n)$ levels of recursion, at most $\eps$ fraction of the nodes are dead throughout the algorithm. Hence the final clustering contains at least $1 - \eps$ fraction of the nodes. Moreover, by the specification of \cref{lem:tool}, each cluster in the final clustering has strong diameter $O(\log^2 n / \eps)$, as required.

\paragraph{Round Complexity}  We only run the algorithm for \cref{lem:tool} on subgraphs with strong diameter 
$R(n, \Theta(\eps / \log n))$, so the algorithm for \cref{lem:tool}  costs $R(n, \Theta(\eps / \log n)) \cdot O(\log n)$ rounds. The main algorithm has $O(\log n)$ levels of recursion, where in each level we invoke the black-box algorithm $\mathcal{A}$ once and the algorithm for \cref{lem:tool} once. Hence the overall round complexity is
\[
O(\log n) \cdot T(n, \Theta(\eps / \log n)) +
O(\log^2 n) \cdot R(n, \Theta(\eps / \log n)). \qedhere
\]
\end{proof}

Similarly, combining \Cref{thm:transformation2} and \Cref{thm:carving}, we obtain the following result.

\begin{theorem}\label[theorem]{thm:carving2}
There is an $O(\log^{10} n/\eps^2)$-round deterministic distributed algorithm in the \congest model that computes a strong-diameter ball carving  of an $n$-node graph $G$ that removes an $\eps$ fraction of the nodes so that each remaining connected component has strong diameter  $D = O(\log^2 n/\eps)$.
\end{theorem}
\begin{proof}
The ball carving algorithm of \Cref{thm:carving} has strong diameter $R(n,\eps) = O(\log^3 n / \eps)$ and round complexity $T(n, \eps) = O(\log^7 n / \eps^2)$. Combining this algorithm with the transformation described in \Cref{thm:transformation2}, we obtain a new ball carving algorithm with strong diameter $D = O(\log^2 n/\eps)$ and round complexity
\ifdefined\DoubleColumn
\begin{align*}
& O(\log n) \cdot T(n, \Theta(\eps / \log n)) +
O(\log^2 n) \cdot R(n, \Theta(\eps / \log n))\\
& = O(\log^{10} n/\eps^2). \qedhere
\end{align*}
\else
\[
O(\log n) \cdot T(n, \Theta(\eps / \log n)) +
O(\log^2 n) \cdot R(n, \Theta(\eps / \log n)) = O(\log^{10} n/\eps^2). \qedhere
\]
\fi
\end{proof}

Again, by the standard reduction from network decompositions to ball carving, we obtain the following theorem.

\begin{theorem}\label[theorem]{thm:decomposition2}
There is an $O(\log^{11} n)$-round deterministic distributed algorithm in the \congest model that computes a strong-diameter network decomposition  of an $n$-node graph $G$ using $C = O(\log n)$ colors and with clusters of diameter $D = O(\log^2 n)$.
\end{theorem}
\begin{proof}
Similar to the proof of \cref{thm:decomposition}, we obtain the desired strong-diameter network decomposition via $O(\log n)$ iterations of the ball carving algorithm of \cref{thm:carving} with  $\eps=1/2$. In each application, we cluster at least half of the nodes into non-adjacent clusters of strong-diameter $O(\log^2 n)$, in $O(\log^{10} n)$ rounds. We then remove all these clustered nodes and repeat on the remaining nodes. In each application, half of the remaining nodes get clustered; those clustered in the $i$th iteration receive color $i$. Hence, after $O(\log n)$ iterations, all nodes are clustered. Thus, we obtain our network decomposition with $O(\log n)$ colors, $O(\log^2 n)$ strong-diameter, and using $O(\log^{11} n)$ rounds.
\end{proof}

\subsection*{Barriers for Further Improvement in the Current Construction}

Coincidentally, both of the deterministic ball carving algorithms in this section and in~\cite{ghaffariGR2021improvedND} achieve the same diameter  bound of $O(\log^2 n / \eps)$. It remains an intriguing open problem to improve this bound for deterministic ball carving, even for weak diameter.

We show that the bound $O(\log^2 n / \eps)$
is the limit of the approach taken in this section, in the following sense. There is a graph with conductance $\Omega(\eps / \log n)$ such that any subgraph with  $\poly(n)$ nodes has diameter $\Omega(\log^2 n /\eps)$. In particular, such a graph has the following properties.
\begin{itemize}
    \item {\bf No Balanced Sparse Cut:} For any two non-adjacent node sets $V_1 \subseteq V$ and $V_2 \subseteq V$ such that $|V_1| \geq n/3$, $|V_2| \geq n/3$, we must have $|V \setminus (V_1 \cup V_2)|  = \Omega(\eps n/ \log n)$. 
    \item {\bf No Large  Small-diameter Component:} For any subset $U \subseteq V$ such that $|U| \geq n/3$, the diameter of $G[U]$ must be $\Omega(\log^2 n /\eps)$.
\end{itemize}

Hence the parameters in \cref{lem:tool} are the best possible. Therefore, to bypass the $O(\log^2 n / \eps)$ barrier, one will need to consider a different approach that can simultaneously grow multiple small-diameter clusters.

\paragraph{Construction}
To construct such a graph, let $n' = O(\eps n / \log n)$, and  take any $n'$-node expander graph $G_1$ with constant maximum degree $\Delta = \Theta(1)$ and constant conductance $\Phi = \Theta(1)$. Then we subdivide each edge into a path of length $\log n / \eps$ to obtain an $n$-node graph $G_2$. It is clear that the new graph $G_2$ has conductance $\Theta(\eps / \log n)$. Moreover, because $G_1$ has $\Delta = \Theta(1)$, any subgraph with $\poly(n)$ nodes must have  diameter $\Omega(\log n)$. Therefore, the construction of $G_2$ ensures that any subgraph with $\poly(n)$ nodes must have  diameter $\Omega(\log^2 n / \eps)$, as required.


\bibliographystyle{alpha}
\bibliography{ref}

\newcommand{\etalchar}[1]{$^{#1}$}
\begin{thebibliography}{HKMT21}

\bibitem[ABCP96]{awerbuch96}
Baruch Awerbuch, Bonnie Berger, Lenore Cowen, and David Peleg.
\newblock Fast network decompositions and covers.
\newblock {\em J.\ of Parallel and Distributed Computing}, 39(2):105--114,
  1996.

\bibitem[AGLP89]{awerbuch89}
Baruch Awerbuch, Andrew~V. Goldberg, Michael Luby, and Serge~A. Plotkin.
\newblock Network decomposition and locality in distributed computation.
\newblock In {\em Proceedings of the 30th IEEE Symposium on Foundations of
  Computer Science (FOCS)}, pages 364--369, 1989.

\bibitem[BKM20]{bamberger2020efficient}
Philipp Bamberger, Fabian Kuhn, and Yannic Maus.
\newblock Efficient deterministic distributed coloring with small bandwidth.
\newblock In {\em Proceedings of the 39th Symposium on Principles of
  Distributed Computing {(PODC)}}, pages 243--252, New York, NY, USA, 2020.
  Association for Computing Machinery.

\bibitem[CCLL21]{CensorCLL21}
Keren {Censor-Hillel}, Yi-Jun Chang, Fran\c{c}ois {Le~Gall}, and Dean
  Leitersdorf.
\newblock Tight distributed listing of cliques.
\newblock In {\em Proceedings of the 32nd Annual {ACM-SIAM} Symposium on
  Discrete Algorithms {(SODA)}}, 2021.

\bibitem[CFG{\etalchar{+}}19]{chang2019coloring}
Yi-Jun Chang, Manuela Fischer, Mohsen Ghaffari, Jara Uitto, and Yufan Zheng.
\newblock The complexity of ({$\Delta$}+1) coloring in congested clique,
  massively parallel computation, and centralized local computation.
\newblock In {\em Proceedings of the 2019 ACM Symposium on Principles of
  Distributed Computing {(PODC)}}, pages 471--480. ACM, 2019.

\bibitem[CGL20]{CensorLL20}
Keren Censor{-}Hillel, Fran{\c{c}}ois~Le Gall, and Dean Leitersdorf.
\newblock On distributed listing of cliques.
\newblock In {\em Proceedings of the {ACM} Symposium on Principles of
  Distributed Computing ({PODC})}, 2020.

\bibitem[CHPS17]{censor2017derandomizing}
Keren Censor-Hillel, Merav Parter, and Gregory Schwartzman.
\newblock Derandomizing local distributed algorithms under bandwidth
  restrictions.
\newblock In {\em Proceedings of the 31st International Symposium on
  Distributed Computing (DISC)}. Schloss Dagstuhl-Leibniz-Zentrum fuer
  Informatik, 2017.

\bibitem[CPP20]{ChatterjeePN20}
Soumyottam Chatterjee, Gopal Pandurangan, and Nguyen~Dinh Pham.
\newblock Distributed {MST}: a smoothed analysis.
\newblock In {\em Proceedings of the 21st International Conference on
  Distributed Computing and Networking {(ICDCN)}}, New York, NY, USA, 2020.
  Association for Computing Machinery.

\bibitem[CPSZ21]{ChangS19}
Yi-Jun Chang, Seth Pettie, Thatchaphol Saranurak, and Hengjie Zhang.
\newblock Near-optimal distributed triangle enumeration via expander
  decompositions.
\newblock {\em J. ACM}, 68(3), 2021.

\bibitem[CS20]{ChangS20}
Yi-Jun Chang and Thatchaphol Saranurak.
\newblock Deterministic distributed expander decomposition and routing with
  applications in distributed derandomization.
\newblock In {\em Proceedings of the 61st Annual {IEEE} Symposium on
  Foundations of Computer Science {(FOCS)}}, 2020.

\bibitem[DHNS19]{Daga2019distributed}
Mohit Daga, Monika Henzinger, Danupon Nanongkai, and Thatchaphol Saranurak.
\newblock Distributed edge connectivity in sublinear time.
\newblock In {\em Proceedings of the 51st Annual {ACM} Symposium on Theory of
  Computing ({STOC})}, pages 343--354, 2019.

\bibitem[EFF{\etalchar{+}}19]{EFFKO19}
Talya Eden, Nimrod Fiat, Orr Fischer, Fabian Kuhn, and Rotem Oshman.
\newblock Sublinear-time distributed algorithms for detecting small cliques and
  even cycles.
\newblock In {\em Proceedings of the International Symposium on Distributed
  Computing (DISC)}, pages 15:1--15:16, 2019.

\bibitem[EN16]{elkin16_decomp}
Michael Elkin and Ofer Neiman.
\newblock Distributed strong diameter network decomposition.
\newblock In {\em Proceedings of the 35th ACM Symposium on Principles of
  Distributed Computing (PODC)}, pages 211--216, 2016.

\bibitem[GGR21]{ghaffariGR2021improvedND}
Mohsen Ghaffari, Christoph Grunau, and V{\'a}clav Rozho{\v{n}}.
\newblock Improved deterministic network decomposition.
\newblock In {\em Proceedings of the {ACM-SIAM} Symposium on Discrete
  Algorithms (SODA)}, pages 2904--2923, 2021.

\bibitem[Gha16]{ghaffari2016MIS}
Mohsen Ghaffari.
\newblock An improved distributed algorithm for maximal independent set.
\newblock In {\em Proceedings of the 25th ACM-SIAM Symposium on Discrete
  Algorithms (SODA)}, pages 270--277, 2016.

\bibitem[Gha19]{ghaffari2019MIS}
Mohsen Ghaffari.
\newblock Distributed maximal independent set using small messages.
\newblock In {\em Proceedings of the 28th ACM-SIAM Symposium on Discrete
  Algorithms (SODA)}, pages 805--820, 2019.

\bibitem[Gha20]{ghaffari2020network}
Mohsen Ghaffari.
\newblock Network decomposition and distributed derandomization.
\newblock In {\em Proceedings of the International Colloquium on Structural
  Information and Communication Complexity {(SIROCCO)}}, pages 3--18. Springer,
  2020.

\bibitem[GHK18]{ghaffari2018derandomizing}
Mohsen Ghaffari, David Harris, and Fabian Kuhn.
\newblock On derandomizing local distributed algorithms.
\newblock In {\em Proceedings of the {IEEE} Symposium on Foundations of
  Computer Science (FOCS)}, pages 662--673, 2018.

\bibitem[GKM17]{ghaffari2017complexity}
Mohsen Ghaffari, Fabian Kuhn, and Yannic Maus.
\newblock On the complexity of local distributed graph problems.
\newblock In {\em Proceedings of the 49th ACM Symposium on Theory of Computing
  (STOC)}, pages 784--797, 2017.

\bibitem[GKS17]{GhaffariKS17}
Mohsen Ghaffari, Fabian Kuhn, and Hsin-Hao Su.
\newblock Distributed {MST} and routing in almost mixing time.
\newblock In {\em Proceedings 37th ACM Symposium on Principles of Distributed
  Computing (PODC)}, pages 131--140, 2017.

\bibitem[GL18]{GhaffariL2018}
Mohsen Ghaffari and Jason Li.
\newblock New distributed algorithms in almost mixing time via transformations
  from parallel algorithms.
\newblock In Ulrich Schmid and Josef Widder, editors, {\em Proceedings 32nd
  International Symposium on Distributed Computing {(DISC)}}, volume 121 of
  {\em Leibniz International Proceedings in Informatics (LIPIcs)}, pages
  31:1--31:16, Dagstuhl, Germany, 2018. Schloss Dagstuhl--Leibniz-Zentrum fuer
  Informatik.

\bibitem[HKMT21]{halldorsson2020efficient}
Magn{\'u}s~M Halld{\'o}rsson, Fabian Kuhn, Yannic Maus, and Tigran Tonoyan.
\newblock Efficient randomized distributed coloring in {CONGEST}.
\newblock In {\em Proceedings of the {ACM} Symposium on Theory of Computation
  (STOC)}, 2021.

\bibitem[IGM20]{izumiLM20}
Taisuke Izumi, Fran{\c{c}}ois~Le Gall, and Fr{\'e}d{\'e}ric Magniez.
\newblock Quantum distributed algorithm for triangle finding in the {CONGEST}
  model.
\newblock In Christophe Paul and Markus Bl{\"a}ser, editors, {\em Proceedings
  of the 37th International Symposium on Theoretical Aspects of Computer
  Science (STACS)}, volume 154 of {\em Leibniz International Proceedings in
  Informatics (LIPIcs)}, pages 23:1--23:13, Dagstuhl, Germany, 2020. Schloss
  Dagstuhl--Leibniz-Zentrum fuer Informatik.

\bibitem[Lin92]{linial92}
Nati Linial.
\newblock Locality in distributed graph algorithms.
\newblock {\em SIAM Journal on Computing}, 21(1):193--201, 1992.

\bibitem[LS93]{linial93}
Nati Linial and Michael Saks.
\newblock Low diameter graph decompositions.
\newblock {\em Combinatorica}, 13(4):441--454, 1993.

\bibitem[MPX13]{miller2013parallel}
Gary~L Miller, Richard Peng, and Shen~Chen Xu.
\newblock Parallel graph decompositions using random shifts.
\newblock In {\em Proceedings of the 25th annual ACM symposium on Parallelism
  in algorithms and architectures {(SPAA)}}, pages 196--203, 2013.

\bibitem[Pel00]{peleg00}
David Peleg.
\newblock {\em Distributed Computing: A Locality-Sensitive Approach}.
\newblock SIAM, 2000.

\bibitem[PS92]{panconesi-srinivasan}
Alessandro Panconesi and Aravind Srinivasan.
\newblock Improved distributed algorithms for coloring and network
  decomposition problems.
\newblock In {\em Proceedings of the 24th ACM Symposium on Theory of Computing
  (STOC)}, pages 581--592, 1992.

\bibitem[RG20]{RozhonG19}
V\'{a}clav Rozho\v{n} and Mohsen Ghaffari.
\newblock Polylogarithmic-time deterministic network decomposition and
  distributed derandomization.
\newblock In {\em Proceedings of the {ACM} Symposium on Theory of Computation
  (STOC)}, 2020.

\bibitem[SV19]{SuVdsic19}
Hsin-Hao Su and Hoa~T. Vu.
\newblock {Distributed data summarization in well-connected networks}.
\newblock In Jukka Suomela, editor, {\em 33rd International Symposium on
  Distributed Computing (DISC)}, volume 146 of {\em Leibniz International
  Proceedings in Informatics (LIPIcs)}, pages 33:1--33:16, Dagstuhl, Germany,
  2019. Schloss Dagstuhl--Leibniz-Zentrum fuer Informatik.

\bibitem[SV20]{su_et_al:LIPIcs:2020:13093}
Hsin-Hao Su and Hoa~T. Vu.
\newblock {Distributed Dense Subgraph Detection and Low Outdegree Orientation}.
\newblock In Hagit Attiya, editor, {\em 34th International Symposium on
  Distributed Computing (DISC)}, volume 179 of {\em Leibniz International
  Proceedings in Informatics (LIPIcs)}, pages 15:1--15:18, Dagstuhl, Germany,
  2020. Schloss Dagstuhl--Leibniz-Zentrum f{\"u}r Informatik.

\end{thebibliography}

\end{document}